\documentclass[conference]{IEEEtran}

\usepackage{cite}
\usepackage{float}
\usepackage{multicol,multirow}
\usepackage{makecell,booktabs}
\usepackage[colorlinks=true,linkcolor=black,anchorcolor=black,citecolor=black,filecolor=black,menucolor=black,runcolor=black,urlcolor=black]{hyperref}
\usepackage{graphicx}
\usepackage{amsfonts}
\usepackage{amssymb}
\usepackage[cmex10]{amsmath}
\usepackage{amsthm}
\usepackage{mathtools}
\usepackage{array}
\usepackage{color}
\usepackage{epstopdf}
\usepackage{bbm}
\usepackage{bm}
\usepackage{comment}
\usepackage{forest}
\usepackage{nomencl}
\usepackage{multirow}
\usepackage{accents}
\usepackage{tabularx}

\newcolumntype{P}[1]{>{\centering\arraybackslash}p{#1}}
\newcolumntype{L}[1]{>{\raggedright\let\newline\\\arraybackslash\hspace{0pt}}m{#1}}
\newcolumntype{C}[1]{>{\centering\let\newline\\\arraybackslash\hspace{0pt}}m{#1}}
\newcolumntype{R}[1]{>{\raggedleft\let\newline\\\arraybackslash\hspace{0pt}}m{#1}}

\theoremstyle{definition} % bold title, normal text
\newtheorem{prop}{Proposition}
\theoremstyle{definition} % bold title, normal text

\theoremstyle{remark} % italic title, normal text

\begin{document}
\title{Causality-based Cost Allocation for\\ 
    Peer-to-Peer Energy Trading in Distribution System}

\author{
\IEEEauthorblockN{Hyun Joong Kim}
\IEEEauthorblockA{Dept. of Energy Engineering \\
Korea Institute of Energy Technology\\
Naju, South Jeolla, Korea
\\hyunjoongkim@kentech.ac.kr\vspace{-12mm}}\\
\and
\IEEEauthorblockN{Yong Hyun Song}
\IEEEauthorblockA{NEXT group\\
Seoul, Korea
\\yh.song@nextgroup.or.kr\vspace{-12mm}}
\and
\IEEEauthorblockN{Jip Kim}
\IEEEauthorblockA{Dept. of Energy Engineering \\
Korea Institute of Energy Technology\\
Naju, South Jeolla, Korea
\\jipkim@kentech.ac.kr\vspace{-12mm}}
}

\maketitle

\begin{abstract}
While peer-to-peer energy trading has the potential to harness the capabilities of small-scale energy resources, a peer-matching process often overlooks power grid conditions, yielding increased losses, line congestion, and voltage problems.
This imposes a great challenge on the distribution system operator (DSO), which can eventually limit peer-to-peer energy trading.
To align the peer-matching process with the physical grid conditions, this paper proposes a cost causality-based network cost allocation method and the grid-aware peer-matching process.
Building on the cost causality principle, the proposed model utilizes the network cost (loss, congestion, and voltage) as a signal to encourage peers to adjust their preferences ensuring that matches are more in line with grid conditions, leading to enhanced social welfare. Additionally, this paper presents mathematical proof showing the superiority of the causality-based cost allocation over existing methods. 
\end{abstract}


\section{introduction}
The power system is undergoing a transition from its traditional hierarchical structure to a more decentralized structure to accommodate the increasing penetration of small-scale energy resources. Among the possible alternatives, peer-to-peer (P2P) energy trading, a promising approach to harnessing the potential of resources, inevitably exploits electric distribution systems~\cite{kim2023pricing}. The integration of P2P energy trading influences the efficiency and reliability of the distribution network managed by the distribution system operator (DSO). Without appropriate coordination of the P2P energy trading platform, the outcomes of P2P energy trading can compromise voltage security, induce line congestion and result in increased system losses\cite{tushar2021peer}.

To address this challenges posed by the integration of P2P energy trading, it is vital to motivate peers to engage in network-friendly trades. 
Recent literature has explored the imposition of network costs to achieve this objective. Universal cost allocation, which charges all peers to pay the network cost proportionally to their trading volume at a fixed rate, widely employed for distributing network costs arising from P2P energy trading \cite{tushar2019grid,kim2019direct,nguyen2021distributed}. However, universal cost allocation is externally set and fails to account for the parties responsible for stressing the grid conditions. 
Alternative cost allocation approaches, such as the electrical distance and zonal methods, have been put forward. Theses methods have shown promise in fostering local energy trading and reduce system loss \cite{anoh2019energy,paudel2020peer,baroche2019exogenous}. Nonetheless, they fall short in encapsulating the causality between P2P trading and the grid conditions.
Another coordination framework, which restrict transactions violating network constraints, is considered as a robust method to ensure network reliability \cite{guerrero2018decentralized,haggi2021multi,morstyn2019integrating}. However, there is a limitation in inducing action of peers to the optimal market efficiency, leading to sub-optimal trading outcomes because of limiting the opportunities for peers to improve market results.
The approach in \cite{kim2019p2p} calculates distribution marginal prices (DLMPs) using trading information acquired through the trading process and allocate network costs. This coordination incentivizes peers to produce network-friendly trading. However, it has a limitation in that it disperses the responsibility for network violations, resulting in unnecessary costs for their resolution.

Causality-based network cost allocation can enhance economic efficiency and competitiveness in the electricity market, and is in compliance with the guidelines of the Federal Power Act (FPA), which mandates that tariffs be just and reasonable \cite{maser2011s}.
% meet the requirement  the Federal Energy Regulatory Commission (FERC) can adopt a just and reasonable tariff based on causality principle \cite{chakraborty2017cost}. 
Despite its significance, the causality principle has largely been overlooked in the context of P2P energy trading. To bridge this gap, we propose a \textit{cost causality}-based network cost allocation method tailored for P2P energy trading. 
This proposed method distributes network costs in a way that mirrors the causal relationship between the actions of individual peers and their subsequent impact on the distribution system.
By doing so, we foster efficient and equitable energy trading among peers, taking into account the physical conditions of the distribution system.
The main contributions of this paper are threefold:
\begin{enumerate}
    \item Demonstrate the superiority of causality-based cost allocation over exogenous methods through mathematical proofs and empirical results.
    \item Propose a causality-based cost allocation method that employs analytically derived partial derivatives from the power flow equation.
    \item Design coordinated P2P energy trading and network operation using causality-based cost allocation.
\end{enumerate}

\section{P2P energy trading in distribution network}
\allowdisplaybreaks
In this section, we propose causality-based network cost allocation and demonstrate the superiority of causality-based network cost allocation.

\subsection{Network description}
Consider a radial distribution network, $\mathcal{D(N,L)}$, consisting of a set of nodes $\mathcal{N}$ and a set of line $\mathcal{L}$. 
Nodes are indexed as ${n}\!=\!0,1,\!\ldots\!,{N}$, where node 0 acts as the slack bus. 
Lines are indexed by ${l}\!=\!0,1,\!\ldots\!,{L}$. Any line $l$ is represented by the pair $(n,m)$ denoting the nodes it connects.
For line $l$, $i_l$ represents the complex current and ${s}^\mathrm{f}_{l} \!=\! {p}^\mathrm{f}_{l} \!+\! \boldsymbol{j}{q}^\mathrm{f}_{l}$ indicates the complex power flow.
$Y_{nm} \!=\! G_{nm} \!+\! \boldsymbol{j}\mathrm{B}_{nm}$ are the admittance of the line connecting nodes $n$ to $m$. 
For each node $n$, $v_n$ denotes the complex voltage and ${s}_{n} \!=\! {p}_{n} \!+\! \boldsymbol{j}{q}_{n}$ symbolizes the net complex power injection. The variable $o$ represents system loss.

\subsection{Peer model}
We categorize peers into two groups: a set of selling peers $\mathcal{A}^\mathrm{S}$ indexed by $i$ and a set of buying peers $\mathcal{A}^\mathrm{B}$ indexed by $j$. The combined set of all peers is denoted as $\mathcal{A} = \mathcal{A}^\mathrm{S} \cup \mathcal{A}^\mathrm{B}$, where peers are indexed as ${k} = 1, \ldots, {K}$, and $p_k$ is the trading volume of peer \textit{k} located at node \textit{k}.

\subsubsection{Selling peer}
The cost function $c_i(\cdot)$ of the selling peer ($i\in\mathcal{A}^\mathrm{S}$) is modeled as a quadratic convex function:\begin{subequations}
    \begin{align}
        & {c}_{i}({p}_{i})=\mathrm{\alpha}_{i} {p}_{i}^2
            +\mathrm{\beta}_{i} {p}_{i}
            +\mathrm{\gamma}_{i},
                \label{eq:cost_ftn}\\
        & \underline{{P}}_{i} \le {p}_{i} \le \overline{{P}}_{i},
            \label{eq:gen_bound}
    \end{align}\label{eq:selling_peer_mdoel}
\end{subequations}
where the parameters of the cost function, $\mathrm{\alpha}_{i}$, $\mathrm{\beta}_{i}$, and $\mathrm{\gamma}_{i}$ are specific to each selling peer, and ${p}_{i}$ is power generation of peer ${i}$. 
The cost function is assumed to be strictly convex, while power generation is constrained between $\underline{{P}}_{i}$ and $\overline{{P}}_{i}$.

\subsubsection{Buying peer}

The utility function $h_j(\cdot)$ of a buying peer $(j\!\in\!\mathcal{A}^\mathrm{B})$ is formulated as a piece-wise quadratic function \cite{samadi2010optimal}:
\begin{align}
        &{h}_j(p_j) \coloneqq
        \begin{dcases}
            \mathrm{\beta}_j p_j-\mathrm{\alpha}_j {p}^2_j 
                &\quad \text{if}~\underline{{P}}_{j} \le {p}_{j} \le \frac{\beta_{j}}{2\alpha_{j}},\\
            \frac{\mathrm{\beta}_j^2}{4 \mathrm{\alpha}_j} 
                &\quad \text{if}~\frac{\beta_{j}}{2\alpha_{j}} \le p_j \le \overline{{P}}_{j},
        \end{dcases}\label{utility ftn}
\end{align}\label{eq:buying_peer_model}
where $\mathcal{\alpha}_{j}$ and $\mathcal{\beta}_{j}$ denote consumer parameters inherent to the utility function when electricity consumption ${p}_{j}$ is bounded between $\underline{{P}}_{j}$ and $\overline{{P}}_{j}$.

Peer modeling based on convex functions, nonempty, and compact sets guarantees that the optimal strategic balance point in transactions between peers, called Nash equilibrium, where a peer cannot gain anything by changing only their own strategy, and the first-order condition between the peer’s utility function and constraints at the equilibrium point is satisfied~\cite{rosen1965existence}.

\subsection{Network cost}
The network costs in this paper pertain to three grid conditions: voltage security, line congestion, and system loss. 
We make the assumption that these network costs are levied in a volumetric manner, proportionally to the amount of contributions (changes in voltage $\hat{v}_n$, line flow $\Hat{s}_l$, and system loss $\Hat{o}$) induced by individual peer trade, which can be formulated:\begin{subequations}
    \begin{align}
        c_v|\Hat{v}_n(P)| & \coloneqq c_v(|v_n(P+P^{0})|-|v_n(P^{0})|), \label{voltage security cost} \\
        c_s|\Hat{s}_l^\mathrm{f}(P)| & \coloneqq c_s(|s_l^\mathrm{f}(P+P^{0})|-|s_l^\mathrm{f}(P^{0})|), \label{congestion cost} \\
        c_o\Hat{o}(P) & \coloneqq c_o(o(P+P^{0})-o(P^{0})), \label{loss cost}
    \end{align}\label{network_cost_def}\end{subequations} where $P \!=\! (p_1, \!\ldots\!, p_k, \!\ldots\!, p_K)$ denotes the state variable of the incremental injected power according to energy trading, $P^{0}$ is the initial state without P2P energy trading.
$c_v$, $c_s$, and $c_o$ are unit network costs for voltage security, line congestion, and system loss.
$c_v$ and $c_s$ are imposed only when violating technical limits as follows:\begin{subequations}
    \begin{align}
        &c_v =
            \begin{dcases}
                0, \quad &\text{if}~~ \underline{{V}}_n \le \vert{v}_n(P+P^{0})\vert \le \overline{{V}}_n,\\
                \hat{c}_v, &\text{otherwise,}\\
            \end{dcases}\label{}\\
        &c_s =
            \begin{dcases}
                0, \quad &\text{if}~~ \underline{{S}}_l^\mathrm{f} \le \vert{s}_l^\mathrm{f}(P+P^{0})\vert \le \overline{{S}}_l^\mathrm{f},\\
                \hat{c}_s, &\text{otherwise,}
            \end{dcases}\label{}
    \end{align}\label{flow and voltage function}\end{subequations} where $\hat{c}_v$ and $\hat{c}_s$ are unit voltage and congestion network costs and the dead bands are set as $[\underline{V}_n,\overline{V}_n]$ and $[\underline{S}^\mathrm{f}_l,\overline{S}^\mathrm{f}_l]$.

\subsection{Causality-based network cost allocation}
We derive a mathematical expression of the impact of energy trades on the distribution system and show how to impose the network costs in \eqref{network_cost_def} based on the causality principle.

\begin{prop}
\label{prop:causality based cost allocation}
The network costs incurred by the P2P energy trading in a distribution network can be quantified as follows:\begin{subequations}\begin{align}
    &\! c_v|\Hat{v}_n(P)|  \approx c_v\sum_{k\in \mathcal{A}}\frac{1}{|v_n(P^0)|}\!\operatorname{Re}\!\bigg(\!\overline{v}_n(P^0)\frac{\partial v_n(P^0)}{\partial p_k}\bigg)p_k, \!\!\label{voltage_violation_approximation} \\
    &\! c_s|\Hat{s}_l^\mathrm{f}(P)|  \approx c_s\sum_{k\in \mathcal{A}}\frac{1}{|s_l^\mathrm{f}(P^0)|}\!\operatorname{Re}\!\bigg(\!\overline{s}_l^\mathrm{f}(P^0)\frac{\partial s_l^\mathrm{f}(P^0)}{\partial p_k}\bigg)p_k, \label{congestion_approximation} \\
    &\! c_o\Hat{o}(P) \approx 2c_o\!\sum_{k\in \mathcal{A}}\operatorname{Re}\!\Bigg(\!\sum_{n=0}^{N}\!\frac{\partial \overline{v}_n(P^0)}{\partial p_k}\!\bigg(\!\sum_{m=0}^{N}\!{G}_{nm}v_m(P^0)\!\bigg)\!\Bigg)p_k. \label{loss_approximation}
\end{align}\label{causality-based allocated network cost}
\end{subequations}
\end{prop}
\vspace{-8mm}
\begin{proof}
\label{proof:causality based cost allocation}
Using the Taylor series expansion, the changes in voltage $(\Hat{v}_n(P))$, line flow $(\Hat{s}_l^\mathrm{f}(P))$, and system loss $(\Hat{o}(P))$ in \eqref{network_cost_def} can be linearly approximated at the initial state $(P^{0})$, neglecting high order terms:
\begin{subequations}
    \begin{align}
        |{v}_n(P)| &\approx  \sum_{k\in\mathcal{A}}\frac{\partial \vert v_n(P) \vert }{\partial p_k}\Big|_{\substack{P=P^0}} p_k + |{v}_n(P^0)|\label{eq:partial_VtoP},\\
        |{s}^\mathrm{f}_l(P)| &\approx  \sum_{k\in\mathcal{A}}\frac{\partial \vert s^\mathrm{f}_l(P) \vert }{\partial p_k}\Big|_{\substack{P=P^0}} p_k +|{s}^\mathrm{f}_l(P^0)|\label{eq:partial_StoP},\\
        {o}(P) &\approx  \sum_{k\in\mathcal{A}}\frac{\partial o(P) }{\partial p_k}\Big|_{\substack{P=P^0}}  p_k+{o}(P^0)\label{eq:partial_OtoP}.
    \end{align}
    \label{eq:partialtoP}
\end{subequations}
The squared voltage magnitude is written as:
\begin{align}
    % & |\Hat{v}_n(P)|^2 = (\vert v_n(P+P^0)\vert  - \vert v_n (P^0)\vert)^2\\
    % & = \vert v_n(P+P^0)\vert^2 + \vert v_n (P^0)\vert^2 - \vert v_n (P+ P^0) \vert \vert  v_n (P^0) \vert\\
    & |{v}_n(P)|^2 = v_n(P)\overline{v}_n(P).
    \label{Voltage magnitude}
\end{align}
By taking partial differentiation with respect to $p_k$ in \eqref{Voltage magnitude}, we can define the partial derivative of the voltage magnitude as:
\begin{align}
\begin{split}
    & \frac{\partial |{v}_n(P)|}{\partial p_k} 
    = \frac{1}{|v_n(P)|}\operatorname{Re}\bigg(\overline{v}_n(P)\frac{\partial v_n(P)}{\partial p_k}\bigg).
    \label{PDE of Voltage magnitude}
\end{split}
\end{align}
We can obtain the partial derivative of the complex voltage $(\!\frac{\partial v_n(P)}{\partial p_k}\!)$ in \eqref{PDE of Voltage magnitude} using the power flow equation. The relationship between injected power and voltage at node \textit{n} is expressed as:
\begin{align}
\begin{split}
    & \overline{s}_n(P) = \overline{v}_n(P) \!\sum_{m=0}^{N}\!{Y}_{nm}v_m(P), \quad \forall n \!\in\! \mathcal{N} \!\setminus\! \{0\},
    \label{pfeqn}
\end{split}
\end{align}
where $\overline{s}_n$ and $\overline{v}_n$ are the conjugate of complex injected power and voltage at node \textit{n}. As known, the slack bus is to maintain its voltage constant. Thus, it holds that:
\begin{align}
    \frac{\partial v_0(P)}{\partial p_k}=0,\frac{\partial \overline{v}_0(P)}{\partial p_k}=0.
    \label{v/p at slack}
\end{align}
By taking partial differentiation with respect to $p_k$ in \eqref{pfeqn}, we obtain the system equation as follows:
\begin{align}
    \mathbbm{1}\!_{\{n\!=\!k\}}\!=\!
        \frac{\partial \overline{v}_n(P)}{\partial p_k} \!\sum_{m=0}^{N} \!{Y}_{nm} v_m(P) +\overline{v}_n(P)\!\sum_{m=1}^{N}\!{Y}_{nm}\frac{\partial v_m(P)}{\partial p_k}.
    \label{PDE_V}
\end{align}
Equation~\eqref{PDE_V} is not linear with respect to complex numbers; however, it exhibits linearity over $\frac{\partial v_n(P)}{\partial p_k}$ and $\frac{\partial \overline{v}_n(P)}{\partial p_k}$. 
Furthermore, it has an unique solution within a radial distribution network, which can be computed by solving \eqref{PDE_V} with rectangular coordinates\cite{christakou2013efficient}. Once the partial derivative of voltage is calculated, the partial derivative of the voltage magnitude can be calculated by substituting the value into \eqref{PDE of Voltage magnitude} as follows:
\begin{align}
    \frac{\partial |{v}_n(P)|}{\partial p_k} = \frac{1}{|v_n(P)|}\!\operatorname{Re}\!\bigg(\!\overline{v}_n(P)\frac{\partial v_n(P)}{\partial p_k}\bigg)
\end{align}

Similarly, the partial derivative of the flow magnitude to the P2P trading in \eqref{eq:partial_StoP} is derived from the power flow equation:
\begin{align}
\begin{split}
    {s_l^\mathrm{f}(P)}={v_n(P)}{\overline{v}_n(P)}{\overline{{Y}}_{nm}}.
    \label{line flow equation}
\end{split}
\end{align}
By taking partial differentiation with respect to $p_k$ in \eqref{line flow equation}, we obtain the partial derivative of complex flow as follows:
\begin{align}
\begin{split}
    \frac{\partial s_l^\mathrm{f}(P)}{\partial p_k}
    &\!=\!\overline{Y}_{nm}\!\bigg(\!\overline{v}_n(P)\!\frac{\partial v_n(P)}{\partial p_k}
    \!+\!v_n(P)\!\frac{\partial \overline{v}_n(P)}{\partial p_k}\!\bigg)\\
    &\!=\!2\overline{Y}_{nm}\operatorname{Re}\bigg(\overline{v}_n(P)\frac{\partial v_n(P)}{\partial p_k}\bigg).
        \label{line flow derivative}
\end{split}
\end{align}
$\frac{\partial s_l^\mathrm{f}(P)}{\partial p_k}$ can be computed by substituting $\frac{\partial v_n(P)}{\partial p_k}$ obtained from \eqref{PDE_V} into \eqref{line flow derivative}. After the value is calculated, the partial derivative of flow magnitude can be obtained by:
\begin{align}
    \frac{\partial |{s}_l^\mathrm{f}(P)|}{\partial p_k}
    =\frac{1}{|s_l^\mathrm{f}(P)|}\operatorname{Re}\bigg(\overline{s}_l^\mathrm{f}(P)\frac{\partial s_l^\mathrm{f}(P)}{\partial p_k}\bigg).
        \label{line flow mag derivative}
\end{align}

Finally, the partial derivative of the system loss to the P2P trading in \eqref{eq:partial_OtoP} can be obtained using a similar approach using the line conductance ${G}_{nm}$ and nodal voltage as follows \cite{zhou2008simplified}:
\begin{align}
    \frac{\partial {o}(P)}{\partial p_k}\!
    &= \!2c_o\!\sum_{k\in \mathcal{A}}\operatorname{Re}\!\Bigg(\!\sum_{n=0}^{N}\!\frac{\partial \overline{v}_n(P)}{\partial p_k}\!\bigg(\!\sum_{m=0}^{N}\!{G}_{nm}v_m(P)\!\bigg)\!\Bigg)p_k,
    \label{take the partial derivative to system loss}
\end{align}
where the derivative of conjugate complex voltage $(\frac{\partial \overline{v}_n(P)}{\partial p_k})$ is obtained from \eqref{PDE_V}.
\end{proof}
Proposition~\ref{prop:causality based cost allocation} facilitates the allocation of network costs according to the contribution of each trade, as measured by the partial derivative. This approach quantifies the impact of a peer trade on changes in voltage, flow, and system loss.

\subsection{Superiority of causality-based network cost allocation}
We demonstrate that the superiority of causality-based cost allocation in achieving optimal trading outcomes over the universal cost allocation in three steps:
1) We first derive the trading results that maximize the social welfare.
2) Next we demonstrate that the optimal market outcome under universal cost allocation cannot be achieved under normal conditions. 3) Finally we prove that the optimal market outcomes under causality-based network cost allocation is achieved at Nash equilibrium state.

\subsubsection{Social welfare maximization}
The outcome of P2P energy trading, where the sum of peer welfare is maximized, can be modeled as follows:\begin{align}\begin{split}
    W^\mathrm{opt}(P)\coloneqq \max_{P\in\mathbb{R}^{K}}
    &\sum_{k\in \mathcal{A}}u_k(p_k)
    -c_v\sum_{n\in N}|\Hat{v}_n(P)|\\
    &
    -c_s\sum_{l\in L}|\Hat{s}_l^\mathrm{f}(P)|
    -c_o\Hat{o}(P),
    \end{split}\label{optimal market equation}
\end{align}
where the revenue and utility functions of a peer are modeled using the generation quantity of selling peer $p_{k}$, the trading volume between peers $p_{ik}$ and the trading price $\mathrm{\lambda}_{k}^*$, which is proposed by the selling peer $k$ and remains constant in the equilibrium state of a perfectly competitive market \cite{yan2020distribution}. The utility of a peer is formulated as a strictly concave function:
\begin{align}
    \label{welfare of peer}
    {u}_k(p_k)=
        \begin{dcases}
            \mathrm{\lambda}_{k}^* p_k-{c}_k(p_k), &\forall k\in \mathcal{A}^\mathrm{S}\\
            \sum_{i\in \mathcal{A}^\mathrm{S}} h_k(p_{ik})-\sum_{i\in \mathcal{A}^\mathrm{S}}\mathrm{\lambda}_{i}^*p_{ik}, &\forall k\in \mathcal{A}^\mathrm{B}.
        \end{dcases}
\end{align}

If there exists a $P^{*o}=(p_1^{*o},\ldots p_k^{*o},\ldots p_K^{*o})$ that leads to the optimal trading result, then \eqref{optimal market equation} satisfies the stationary condition for every peer as follows:
\begin{align}
    \begin{split}
    \frac{\partial w_k(p_k^{*o})}{\partial p_k}\!=\!{c_v}\!\sum_{n\in \mathcal{N}}\!\frac{\partial \!|\Hat{v}_n(\!P^{*\!o}\!)|}{\partial p_k} \!+\! {c_s}\!\sum_{l\in \mathcal{L}}\!\frac{\partial |\Hat{s}_l^\mathrm{f}(\!P^{*\!o}\!)|}{\partial p_k}
    \!+\!{c_o}\!\frac{\partial \Hat{o}(\!P^{*\!o}\!)}{\partial p_k}.
    \label{stationary condition}
    \end{split}
\end{align}

\subsubsection{Nash equilibrium under universal network cost allocation}
Universal network cost allocation aims to equitably distribute network costs among all peers. The welfare function of a peer in P2P energy trading can be expressed as:
\begin{align}
    \begin{split}
     w^\mathrm{u}_k(P)&\coloneqq  u_k(p_k)
    -{c_v}\sum_{n\in \mathcal{N}}|\hat{v}_n(P)|\frac{p_k}{\sum_{t\in \mathcal{A}}p_t}\\
    & -{c_s}\sum_{l\in \mathcal{L}}|\hat{s}_l^\mathrm{f}(P)|\frac{p_k}{\sum_{t\in \mathcal{A}}p_t}
    -{c_o}\Hat{o}(P)\frac{p_k}{\sum_{t\in \mathcal{A}}p_t}.
    \end{split}\label{welfare function under universal network allocation}
\end{align}
Using the welfare function in \eqref{welfare function under universal network allocation}, we aim to verify if the optimal trading result is obtained at Nash equilibrium state under universal network cost allocation. Thus, we propose and prove:
\begin{prop}
Let $P^{*u}=(p_1^{*u},\ldots p_k^{*u},\ldots,p_K^{*u})$ and $W^\mathrm{u}(P)=\sum_{k\in \mathcal{A}}w^\mathrm{u}_k(P)$ denote the Nash equilibrium state and total social welfare under universal network cost allocation. Then, $W^\mathrm{u}(P^{*u}) \le W^\mathrm{u}(P^{*o})$ and the equality holds when there are only one selling node and one buying node, \textit{i.e.,} all buying peers are located at one node and selling peers are at another node.
\end{prop}
\begin{proof}
Given the the welfare function of a peer under universal network cost allocation in \eqref{welfare function under universal network allocation}, the total social welfare can be represented as:
\begin{align}
\begin{split}
   W^\mathrm{u}(&P) \!=\! \sum_{k\in \mathcal{A}}\!u_k(p_k)\! - \!{c_v}\!\sum_{k\in \mathcal{A}}\!\sum_{n\in \mathcal{N}}\!|\hat{v}_n(P)|\!\frac{p_k}{\sum_{t\in \mathcal{A}}\!p_t} \\
    \!-\!&{c_s}\!\sum_{k\in \mathcal{A}}\!\sum_{l\in \mathcal{L}}\!|\hat{s}_l^\mathrm{f}(P)|\!\frac{p_k}{\sum_{t\in \mathcal{A}}\!p_t}\! - \!{c_o}\!\sum_{k\in \mathcal{A}}\!\Hat{o}(P)\!\frac{p_k}{\sum_{t\in \mathcal{A}}\!p_t}.
\end{split}
\label{total social welfare under uni}
\end{align}
Upon canceling out $\sum_{k\in \mathcal{A}}p_k$ with ${\sum_{t\in \mathcal{A}}p_t}$, $W^\mathrm{u}(P)$ aligns with the social maximization formulation in \eqref{optimal market equation}. Consequently, we can infer:
\begin{align}
    W^\mathrm{u}(P^{*o}) = W^\mathrm{opt}(P^{*o}).
    \label{Eq:comparing welfare under uni_1}
\end{align}
Since $W^\mathrm{u}(P^{*o})$ represents the upper bound of $W^\mathrm{u}(P)$, it follows that $W^\mathrm{u}(P^{*u})$ is less than or equal to $W^\mathrm{u}(P^{*o})$:
\begin{align}
    W^\mathrm{u}(P^{*u}) \leq W^\mathrm{u}(P^{*o}) = W^\mathrm{opt}(P^{*o}).
    \label{Eq:comparing welfare under uni_2}
\end{align}

To delineate the conditions under which equality might hold,assume $P^{*u}$ and $P^{*o}$ are equal. By this assumption, the first derivative of the welfare function at these states should also be equivalent, expressed as:
\begin{align}
    \frac{\partial u_k(p_k^{*u})}{\partial p_k}=\frac{\partial u_k(p_k^{*o})}{\partial p_k},
    \label{universal: derivative of welfare}
\end{align}
where $\frac{\partial u_k(p_k^{*o})}{\partial p_k}$ is previously defined in \eqref{stationary condition}.
Additionally, $\frac{\partial u_k(p_k^{*\!u})}{\partial p_k}$ can be deduced using the first-order condition in \eqref{welfare function under universal network allocation}, with the detailed derivation available in the Appendix.
Substituting these terms into \eqref{universal: derivative of welfare} yields:
\begin{align}
\begin{split}
    &c_v\bigg(\frac{\partial |\Hat{v}_n(P^{*u})|}{\partial p_k}\!-\!\frac{|\Hat{v}_n\!(P^{*u})|}{\sum_{t\in \mathcal{A}}p_t^{*u}}\bigg) 
    \!+\!c_s\bigg(\frac{\partial |\Hat{s}_l^\mathrm{f}\!(P^{*u})|}{\partial p_k}\!-\!\frac{|\Hat{s}_l^\mathrm{f}\!(P^{*u})|}{\sum_{t\in \mathcal{A}}p_t^{*u}}\bigg) \\
    &+c_o\bigg(\frac{\partial \Hat{o}(P^{*u})}{\partial p_k}\!-\!\frac{\Hat{o}(P^{*u})}{\sum_{t\in \mathcal{A}}p_t^{*u}}\bigg)=0.
\end{split}\label{special_case}
\end{align}
Consequently, $W^\mathrm{u}(P^{*u})$ is equal to $W^\mathrm{u}(P^{*o})$ only when condition \eqref{special_case} is satisfied. This scenario where \eqref{special_case} holds is notably rare in a standard distribution system. Specifically, it implies a situation where the marginal voltage magnitude equals to the average voltage magnitude across the trading volume, the marginal flow magnitude matches the average flow magnitude, and the marginal system loss corresponds to the average system loss. This setting essentially posits that all selling peers are located in one node, and all buying peers are clustered at another node within the network.
\end{proof}

\subsubsection{Nash equilibrium under causality-based network cost allocation}
According to Proposition~\ref{prop:causality based cost allocation} the contribution to the network based on the trading volume of each peer can be differentiated, and the welfare function of peer under causality-based cost allocation is modeled as:
\begin{align}
    \begin{split}
            w^\mathrm{c}_k(p_k)\coloneqq &u_k(p_k)
    -{c_v}\sum_{n\in \mathcal{N}}\frac{\partial |{v}_n(P^0)|}{\partial p_k}p_k\\
    &-{c_s}\sum_{l\in \mathcal{L}}\frac{\partial |{s}_l^\mathrm{f}(P^0)|}{\partial p_k}p_k
    -{c_o}\frac{\partial o(P^0)}{\partial p_k}p_k.        
    \end{split}\label{welfare function under causality-based network allocation}
\end{align}
The welfare function in \eqref{welfare function under causality-based network allocation} is employed to ascertain whether the optimal trading result is achieved at Nash equilibrium state under causality-based network cost allocation. Thus, we propose and prove:
\begin{prop}\label{prop3}
Let $P^{*c}\!\!=\!(p_1^{*c},\ldots p_k^{*c},\ldots,p_K^{*c})$ and $W^\mathrm{c}(P)\!=\!\sum_{k\in \mathcal{A}}w^\mathrm{c}_k(P)$ be Nash equilibrium state and total social welfare under causality-based network cost allocation. Then, $W^\mathrm{c}(P^{*c})$ is equal to $W^\mathrm{c}(P^{*o})$.
\end{prop}
\begin{proof}
The first-order condition for \eqref{welfare function under causality-based network allocation} at $P^{*c}$ can be expressed as follows:
\begin{align}
\begin{split}
    &\frac{\partial {u_k}\!(p^{*\!c}_k)}{\partial p_k}
    \!=\!{c_v}\!\sum_{n\in \mathcal{N}}\!\frac{\partial |{v}_n\!(\!P^0\!)|}{\partial p_k}
    \!+\!{c_s}\!\sum_{l\in \mathcal{L}}\!\frac{\partial |{s}_l^\mathrm{f}(\!P^0\!)|}{\partial p_k}
    \!+\!{c_o}\!\frac{\partial o(\!P^0\!)}{\partial p_k}.
    \label{first order condition causality allocation}
\end{split}
\end{align}
By comparing the stationary condition of \eqref{optimal market equation} with \eqref{first order condition causality allocation} and using linear approximation in \eqref{eq:partialtoP}, we obtain that:
\begin{align}
\begin{split}
    &\frac{\partial u_k(p_k^{*o})}{\partial p_k}\\
    &\!=\! {c_v}\!\sum_{n\in \mathcal{N}}\!\frac{\partial |\Hat{v}_n(P^{*o})|}{\partial p_k} \!+\! {c_s}\!\sum_{l\in \mathcal{L}}\!\frac{\partial |\Hat{s}_l^\mathrm{f}(P^{*o})|}{\partial p_k} \!+\! {c_o}\!\frac{\partial \Hat{o}(P^{*o})}{\partial p_k} \\
    &\!=\! {c_v}\!\sum_{n\in \mathcal{N}}\!\frac{\partial |{v}_n(P^0)|}{\partial p_k} \!+\! {c_s}\!\sum_{l\in \mathcal{L}}\!\frac{\partial |{s}_l^\mathrm{f}(P^0)|}{\partial p_k} \!+\! {c_o}\!\frac{\partial o(P^0)}{\partial p_k}\\
    &\!=\! \frac{\partial u_k(p_k^{*c})}{\partial p_k}
\label{optimal_equal_nash}
\end{split}
\end{align}
According to \eqref{welfare of peer}, the welfare function of the peer is strictly concave. Equation~\eqref{optimal_equal_nash} implies that $P^{*c}$ and $P^{*o}$ are equivalent, and the total social welfare under causality-based network cost allocation is equal at these states:
\begin{align}
    W^\mathrm{c}(P^{*c}) = W^\mathrm{c}(P^{*o}).
    \label{Eq:comparing welfare under cau_1}
\end{align}

\begin{figure}[t]
    \vspace{-2mm}
    \centering    
    \includegraphics[width=1.0\columnwidth]{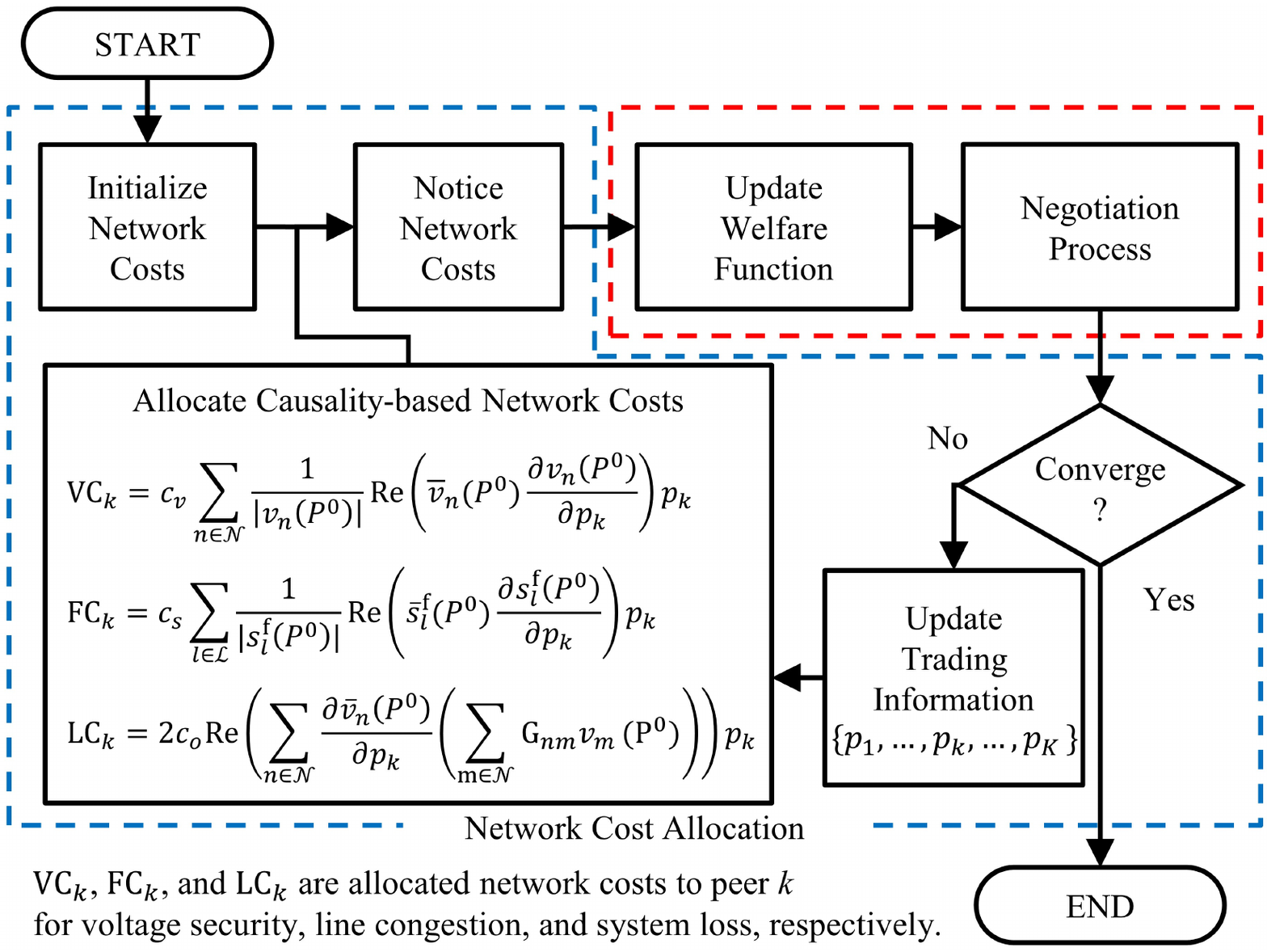}
    \caption{Coordination between P2P energy trading and distribution system using causality-based network cost}\vspace{0mm}
    \label{fig:flow_diagram}    
    \vspace{-4mm}
\end{figure}

Given the the welfare function in \eqref{welfare function under causality-based network allocation}, total social welfare can be formulated as follows:
\begin{align}
\begin{split}
    W^\mathrm{c}(P) &= \sum_{k\in \mathcal{A}}u_k(p_k) - {c_v}\sum_{k\in \mathcal{A}}\sum_{n\in \mathcal{N}}\frac{\partial |{v}_n(P^0)|}{\partial p_k}p_k \\
    -&{c_s}\sum_{k\in \mathcal{A}}\sum_{l\in \mathcal{L}}\frac{\partial |{s}_l^\mathrm{f}(P^0)|}{\partial p_k}p_k - {c_o}\sum_{k\in \mathcal{A}}\frac{\partial o(P^0)}{\partial p_k}p_k.
\end{split}
\label{total social welfare under cau}
\end{align}
Using linear approximation in \eqref{eq:partialtoP}, $W^\mathrm{c}(P)$ is reformulated as social maximization model in \eqref{optimal market equation}. Thus, we can derive that:
\begin{align}
    W^\mathrm{c}(P^{*c}) = W^\mathrm{c}(P^{*o}) = W^\mathrm{opt}(P^{*o}).
    \label{Eq:comparing welfare under cau_2}
\end{align}
\end{proof}
\vspace{-2mm}
Proposition~\ref{prop3} implies that causality-based network cost allocation ensures the optimal P2P energy trading outcomes. It is important to note, however, that network costs for voltage security and line congestion are imposed when network violations occur. Conversely, as system loss costs always occur once a transaction is made, market efficiency is substantially affected by system loss costs.

\subsection{Coordination of P2P energy trading with costs allocation}
We propose coordination of the P2P energy trading with causality-based cost allocation. This approach leans towards a negotiation-oriented trading mechanism, which facilitates welfare adjustments based on real-time network costs. 
As illustrated in Fig.~\ref{fig:flow_diagram}, prior to initiating P2P energy trading, the DSO establishes the initial network costs. Subsequently, this cost data is communiated to all participating peers. Informed by these costs, peers strategically build their initial trading plans and partake in energy trades with other peers.
% illustrates the coordination between the causality-based cost allocation and P2P energy trading. Before starting P2P energy trading, the DSO set the initial network costs and notifies peers of the network costs. Based on this information, peers set their initial trading strategies and start engaging in energy trading with other peers.
During the negotiation process, the DSO collects the trading information $(P)$, and allocates the network costs $\mathrm{VC}_k, \mathrm{FC}_k, \mathrm{LC}_k$ following Proposition~\ref{prop:causality based cost allocation}.
As these costs are updated, peers recalibrate their welfare functions to reflect the changes as follows \cite{paudel2020peer}:
\begin{subequations}
    \begin{align}
        &p^{\tau+1}_k \!= \!\underset{p_k}{\mathrm{argmax}} \{u_k(p_k)\!-\!\mathrm{VC}^{\tau}_{k}\!-\!\mathrm{FC}^{\tau}_{k}\!-\!\mathrm{LC}^{\tau}_{k}\}, \hspace{0mm} \forall k\in \mathcal{A},\!\! \label{update_welfare_1} \\
        &\lambda^{\tau+1}_i = \Big[\lambda^{\tau}_i + \epsilon (\sum_{j\in \mathcal{A^\mathrm{B}}}p^{\tau+1}_{ji} - p^{\tau}_i)\Big]^+\!, \;\! \hspace{0mm}\forall i\in \mathcal{A}^\mathrm{S}, \label{update_welfare_2} \\
        & \mathrm{VC}^{\tau+1}_{k} = \Big[\mathrm{VC}^{\tau}_{k} + c_v\sum_{n\in \mathcal{N}} \mathrm{\Phi}_{n,k}p^{\tau+1}_k\Big]^+, \hspace{0mm} \forall k\in \mathcal{A},\!\! \label{update_welfare_3} \\
        & \mathrm{FC}^{\tau+1}_{k} = \Big[\mathrm{FC}^{\tau}_{k}+c_s\sum_{l\in \mathcal{L}}\mathrm{X}_{l,k}   p^{\tau+1}_k\Big]^+, \hspace{0mm} \forall k\in \mathcal{A},\!\! \label{update_welfare_4} \\ 
        & \mathrm{LC}^{\tau+1}_{k} = \Big[\mathrm{LC}^{\tau}_{k} + 2c_o \mathrm{\Psi}_{k} p^{\tau+1}_k\Big]^{+}, \hspace{0mm} \forall k\in \mathcal{A}.
        \label{update_welfare_5}
    \end{align}\label{update_welfare_ftn}
\end{subequations}
$\epsilon$ and $\tau$ are the tuning parameter and negotiation step. $[x]^+$ represent max(0,x). $\mathrm{\Phi}_{n,k},\mathrm{X}_{l,k}$ and $\mathrm{\Psi}_{k}$ are causal relationship factors between trading volume of peer $k$ in terms of voltage at node $n$ , flow at line $l$ and system loss, with the detailed derivation in the Appendix. Using the updated rule from~\eqref{update_welfare_2} to ~\eqref{update_welfare_5}, peers re-engage in negotiations for energy trading with the updated trading volume using~\eqref{update_welfare_1}. This entire process is iteratively continued until convergence in trading results is achieved.

\begin{figure}[b]
    \vspace{-3mm}
    \centering    
    \includegraphics[width=1.0\columnwidth]{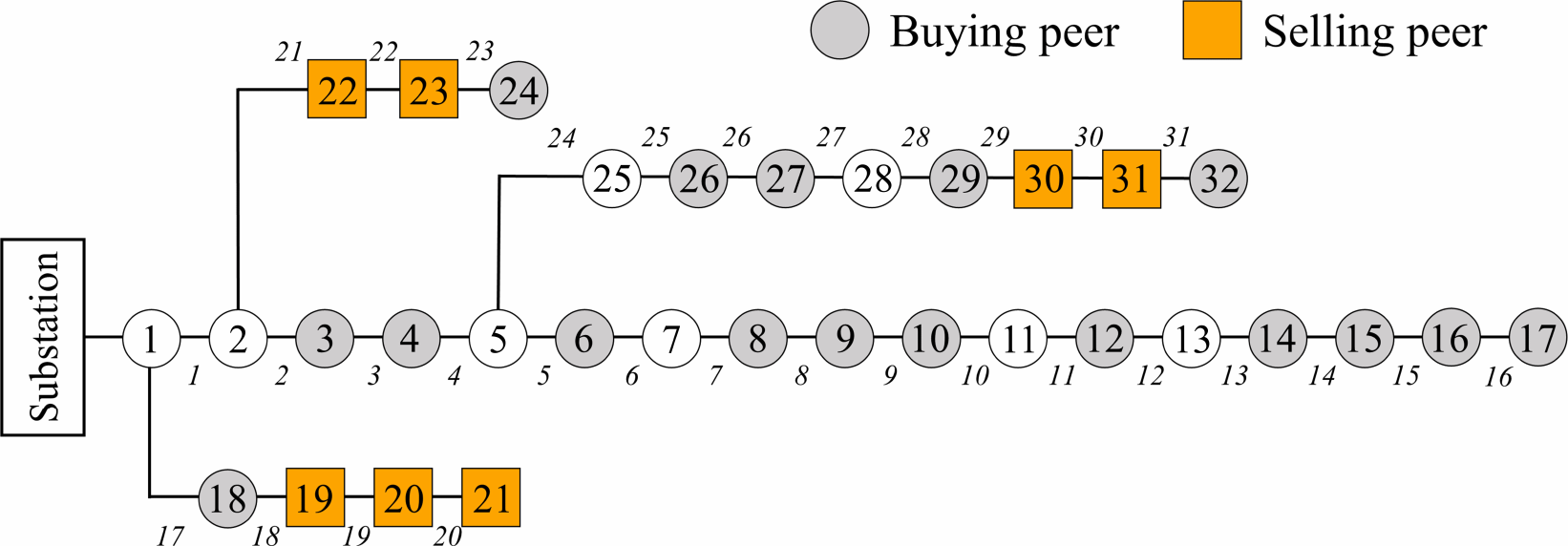}
    \caption{A modified IEEE 33-node distribution network with nodal indices in circles and squares. Italic numbers near edges represent line indices. Grey and orange nodes represent buying peers and selling peers. \vspace{0mm}}
    \label{mod IEEE 33}
    \vspace{-0mm}
\end{figure}
\begin{table}[b]
    \vspace{-3mm}
    \centering
    \caption{\vspace{0mm}Comparison between the base case, universal and causality-based network cost allocation policy in the total trading volume, social welfare, and system loss\vspace{-2mm}}
        \vspace{-0mm}
        \begin{center}
        \begin{tabular}{L{2cm} | C{2cm} C{1.5cm} C{1.5cm}}
            \toprule
            % \centering
            Policy Configuration & Total trading volume $[\mathrm{MWh}]$ & Social welfare [\$] & System loss $[\mathrm{MWh}]$ \\
            \midrule
            Base case  & 4.67 & 581.25 & 0.39 \\
            Universal & 4.50 & 580.60 & 0.36 \\
            Causality-based & 4.53 & 592.59 & 0.28 \\
            \bottomrule
        \end{tabular}
        \end{center}
        \label{t1:confi_loss}
        \vspace{-3mm}
\end{table}

\section{Case study}
The case study uses the modified 33-node distribution system as shown in Fig.~\ref{mod IEEE 33} \cite{dolatabadi2020enhanced} to assess the efficacy of our proposed coordination mechanism for P2P energy trading with causality-based cost allocation. 
For the modeling of the negotiatino process, we have employed the dual gradient method, as detailed in \cite{paudel2020peer}. All simulations were implemented using Python, and the comprehensive source code along with input data is available at \cite{codeanddata}.

\subsection{Simulation settings}
Seven selling peers have the capability to generate electricity with a maximum average installed capacity of 0.7$\mathrm{MW}$. Meanwhile, seventeen buying peers possess a maximum average demand of 0.38$\mathrm{MW}$. The peers engage in energy trading, determining their interactions based on the trading price and volumes for the energy to be delivered over a span of one hour. The configurations established for for this simulation include:
\begin{enumerate}
    \item Base case: No network costs are allocated to the peers.
    \item Universal policy: Peers pay the network costs in proportion to their trading volume at an uniform rate.
    \item Causality-based policy: Network costs are allocated based on causality principle.
\end{enumerate}
In Scenario 1, the focus is primarily on system loss cost while neglecting network constraints. In Scenario 2, we assess the effectiveness of the causality-based policy in achieving grid-aware P2P trading.

\begin{figure}[t]
    \vspace{-8mm}
    \centering    
    \includegraphics[width=1.03\columnwidth]{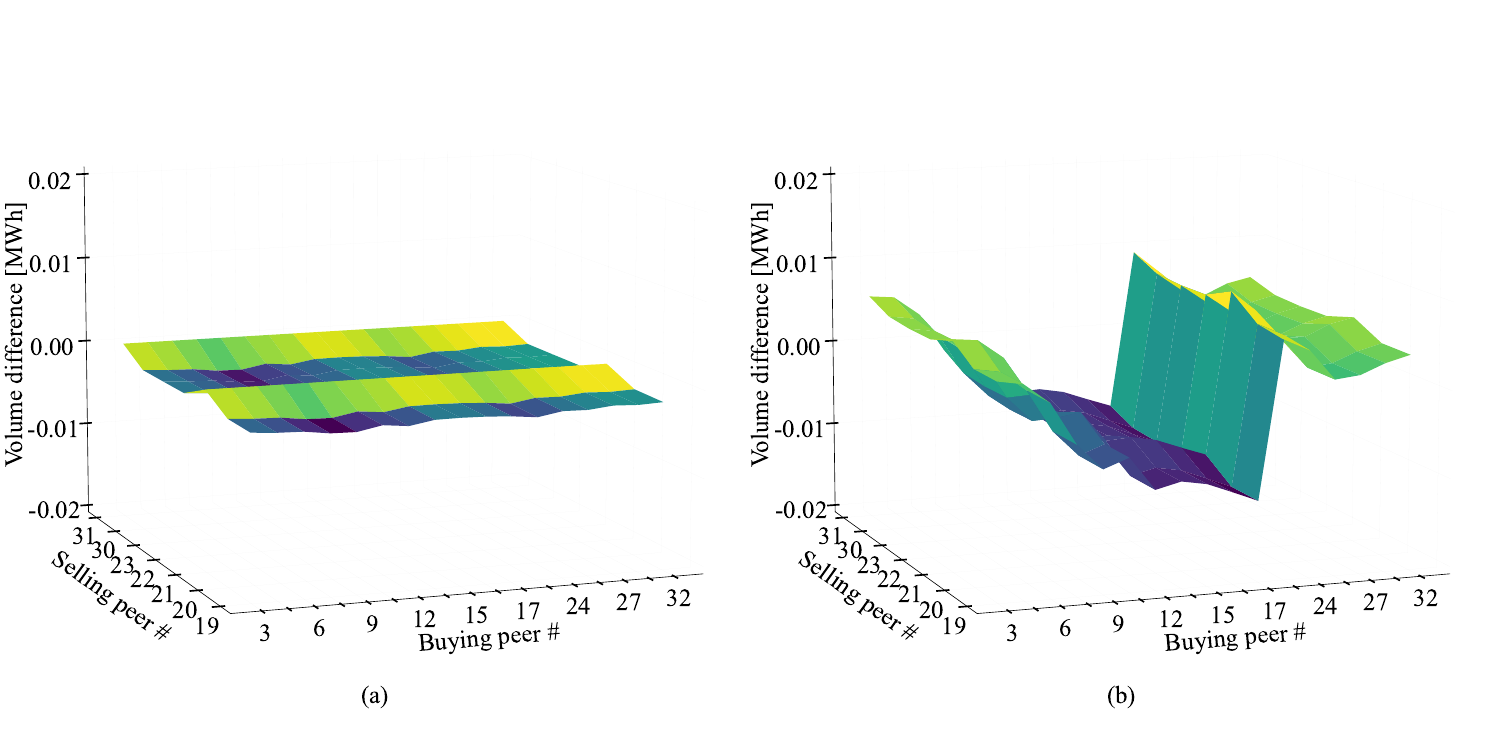}
    \vspace{-8mm}
    \caption{\vspace{-0mm}(a) Difference in trading volumes of peers between the universal policy and the base case, (b) Difference in trading volumes of peers between the causality-based policy and the base case.\vspace{0mm}}
    \label{fig:loss allocation}
    % \vspace{-2mm}
% \end{figure}
% \begin{figure}[htb!]
    \vspace{3mm}
    \centering    
    \includegraphics[width=1.0\columnwidth]{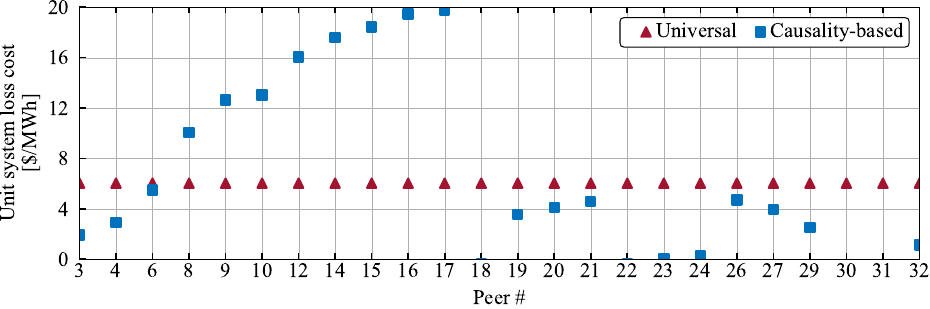}
    \vspace{-7mm}
    \caption{The unit system loss cost charged to peers.\vspace{0mm}}
    \label{fig:Unit system loss cost}    
    \vspace{-2mm}
\end{figure}

\subsection{Scenario 1: P2P platform with loss cost}

Table~\ref{t1:confi_loss} provides a comparison of the total trading volume, social welfare, and system loss across the three configurations. Compared to the base case, both the universal and causality-based policies result in minor reductions in total trading volume, by $0.17\mathrm{MWh}$ and $0.14\mathrm{MWh}$, respectively. In terms of social welfare, defined as the aggregate welfare of all peers minus the system loss cost, the causality-based policy exhibits the highest value. As the trading volume diminishes, a corresponding decrease in system loss is observed under both policies, resulting in improved system and market efficiency, quantified in terms of total trading volume.

Compare to the universal policy, the causality-based policy shows that the trading volume increases by $0.56\%$, but system loss decreases by $22.85\%$, which attributes changes in trading volume of each peer. 
Fig.~\ref{fig:loss allocation} illustrates the differences in trading volume among peers between the two policies and the base case. Under the universal policy, the trading volume of buying peers decrease on average by $0.0099\mathrm{MWh}$, and that of selling peers decrease by $0.024\mathrm{MWh}$. 
The variation in different trading volumes among peers is due to the distinct utility and cost coefficients in \eqref{eq:selling_peer_mdoel} and \eqref{utility ftn} characterizing each peer's welfare function. 
Peer 9, who has the smallest coefficients, experiences the most significant reduction in trading volume. 
In contrast, peer 27 who sees with the least reduction in trading volume has the largest coefficients. 
Even though there is a total decline in trading volume for the causality-based policy compared to the base case, not all peers experience a decrease.
% Despite the total trading volume decrease in the causality-based policy compared to the base case, the trading volume has not decreased for all peers. 
Namely, buying peers 3, 4, 6, 18, 24, 26, 27, 29, and 32 exhibit an average increase in trading volume by $0.029\mathrm{MWh}$. Similarly, selling peer 21 and 31 present nearly unchanged trading volume.

\begin{figure}[b!]
    \vspace{-4mm}
    \centering    
    \includegraphics[width=1.0\columnwidth]{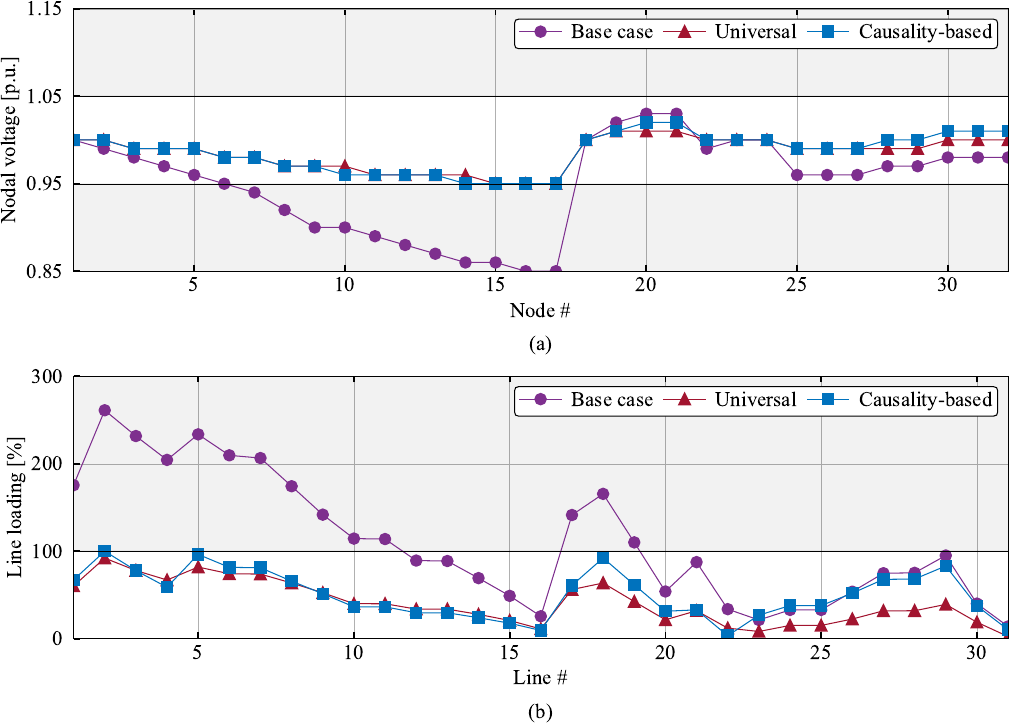}
    \vspace{-6mm}
    \caption{\vspace{0mm}(a) Distribution line loading(in \% relative to flow limit) and (b) Nodal voltage magnitudes under the base case, universal and causality-based policy\vspace{0mm}}
    \label{fig:vollineprofile}    
    % \vspace{-2mm}
\end{figure}
\begin{table}[b!]
    \vspace{-2mm}
    \centering
    \vspace{-2mm}
    \caption{\vspace{-0mm} Trading volume, Social Welfare, System loss, Max/Min of Nodal Voltage, and Average Line Loading\vspace{-2mm}}
        \vspace{-2mm}
        \begin{center}
        \begin{tabular}{L{2cm} | C{0.9cm} C{0.9cm} C{0.9cm} C{0.9cm} C{0.9cm}}
            \toprule
            \centering
            Policy Configuration & Total trading volume $[\mathrm{MWh}]$ & Social welfare $[\$]$ & System loss $[\mathrm{MWh}]$ & Average voltage margin 
            $[\mathrm{p.u.}]$  & Average line loading margin $[\%]$ \\
            \midrule
            Base case & 4.67 & 581.25 & 0.39 & -0.005 & -8.23\\
            Universal & 1.76 & 351.14 & 0.049 & 0.032 & 59.19\\
            Causality-based  & 2.97 & 498.35 & 0.074  & 0.031 & 49.41\\
            \bottomrule
        \end{tabular}
        \end{center}
        \label{t2}
        \vspace{-3mm}
\end{table}

The reason for the difference in trading volume between peers is due to the allocated system loss cost. In the universal policy, a unit system loss cost of $\$6.04/\mathrm{MWh}$ is equally imposed as shown in Fig.~\ref{fig:Unit system loss cost}. However, in the causality-based policy, buying peers 8 to 17 are charged a larger value than the average unit system loss cost of $\$6.64/\mathrm{MWh}$, leading to a reduction in their trading volumes. On the other hand, selling peers 21 and 31, who have small incremental generation cost, are not influenced by the allocated system loss cost.

\begin{figure}[t]
    \vspace{-8mm}
    \centering    
    \includegraphics[width=1.03\columnwidth]{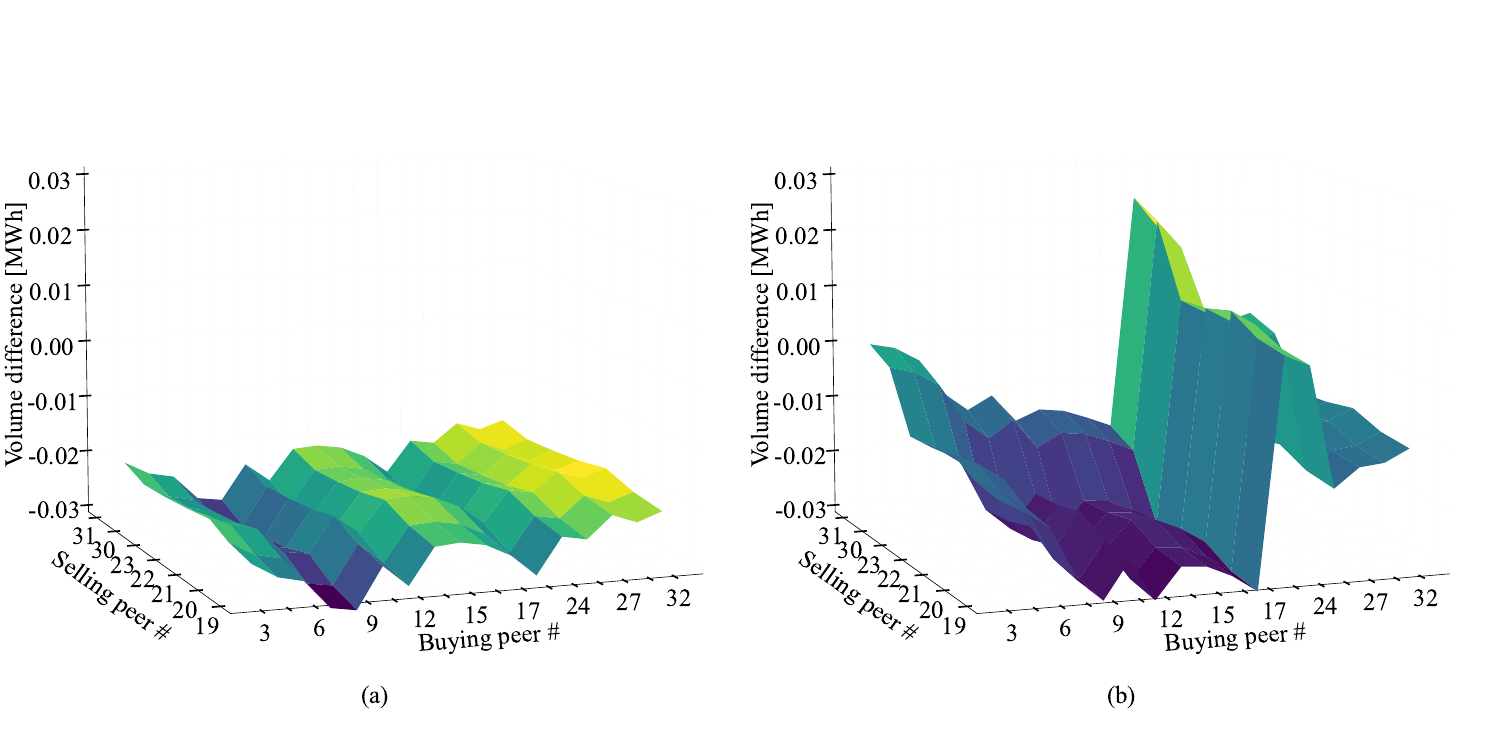}
    \vspace{-8mm}
    \caption{(a) Difference of peer's trading volumes between universal policy and the base case, (b) Difference of peer's trading volumes between causality-based policy and the base case \vspace{0mm}}
    \label{fig:diff_tradingvolume}    
    \vspace{4mm}
    \centering    
    \includegraphics[width=1.0\columnwidth]{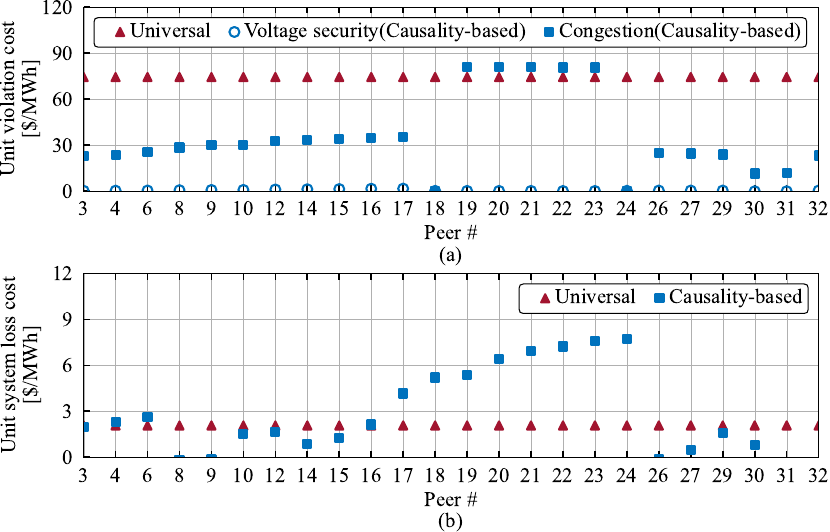}
    \vspace{-6mm}
    \caption{(a) The unit network costs charged to peers in terms of voltage security and congestion, (b) The unit system loss cost charged to peers.\vspace{-2mm}}
    \label{fig:unit cost diagram}
    \vspace{-2mm}
\end{figure}

\subsection{Scenario 2: P2P platform with all network costs}
Table~\ref{t2} shows trading results in terms of system loss, voltage security, and congestion. Compared to the base case, the total trading volume in the universal and causality-based policies decrease by $1.76\mathrm{MWh}$ and $2.97\mathrm{MWh}$. Social welfare and system loss decreased in proportion to total trading volume. In the base case, voltage drop occurs at nodes 7 to 17. In addition, congestion occurs in the lines 1 to 11 and 17 to 19 as shown in Fig.~\ref{fig:vollineprofile}. However, the two policies shows that the voltage magnitude and line loading exist in the security region. The causality-based policy has lower values for the average voltage and line loading margin. This implies that the causality-based policy allows the network used to its limits.

Fig.~\ref{fig:diff_tradingvolume} depicts the difference of trading volume resulting from the two policies comparing to the base case. In the universal policy, the trading volumes of buying peers and selling peers decrease on average by $0.17\mathrm{MWh}$ and $0.42\mathrm{MWh}$. Similar to the universal policy in scenario 1, peer 9 shows the most decrease in trading volume, while peer 27 has the least reduction of trading volume. Peer 19 has the most significant decrease in trading volume, while peer 21 and 31 show the least reduction in trading volume. 
Thus, the consistent decrease in trading volume of all peers contributes to the prevention of constraint violations. Whereas, the change in trading volume among peers varies depending on its impact on violation in the causality-based policy. The trading volume of peer 18 and 24 increase by $0.088\mathrm{MWh}$ and $0.066\mathrm{MWh}$. These increases in trading volume contribute to mitigating congestion on lines 1 to 11 and 17. The trading volume of peers 16 and 17, which directly influences the voltage drop at nodes 16 and 17, decreases by $0.18\mathrm{MWh}$ and $0.19\mathrm{MWh}$. Further, the decreased trading volumes of selling peers 30 and 31 contributes to alleviating congestion on lines 1 to 4.

In the universal policy, whenever violation is anticipated during the trading process, the DSO accumulates the unit violation cost and charges it to peers at a rate of $\$74.40/\mathrm{MWh}$ as shown in Fig.~\ref{fig:unit cost diagram}(a). The unit system loss cost is determined to be $\$2.08/\mathrm{MWh}$ as shown in Fig.~\ref{fig:unit cost diagram}(b). Hence, there is a consistent reduction in trading volume among peers. However, the causality-based policy distinguishes violation costs according to the impact of peers' actions on the network. As shown in Fig.~\ref{fig:unit cost diagram}(a), the DSO can prevent congestion by imposing high unit congestion costs on peers 19 to 23, who could directly influence on lines 1 to 11. Comparing to congestion costs, voltage costs are set to be relatively low because the congestion costs are sufficient to prevent voltage violation. Peer 17 is allocated the unit congestion cost of $\$35.04/\mathrm{MWh}$, whereas the unit voltage security cost is only $\$1.48/\mathrm{MWh}$. The unit system loss cost increases to the average of $\$2.72/\mathrm{MWh}$ compared to the universal policy, as the overall trading volume increases.

\section{conclusion}
This paper proposes a causality-based cost allocation method for the P2P energy trading and distribution network operation. The proposed cost allocation method charges network costs based on the causal relationship between the trading volume of individual peers and their impact on the distribution system.
We demonstrates the superiority of causality-based cost allocation over universal network costs allocation by showing that P2P energy trading results derived under the causality-based cost allocation are suitable for achieving the optimal trading outcome. Coordination between P2P energy trading and the causality-based cost allocation is presented. Simulation results highlight the enhanced energy efficiency and increased social welfare of the coordinated P2P energy trading compared to the universal cost allocation. This approach is also verified to be effective in ensuring voltage security and to prevent line congestion.

\section*{Acknowledgement}
This work was supported by the National Research Foundation of Korea (NRF) grant funded by the Korean government (MSIT) (No. RS-2023-00210018).

% \onecolumn
\appendix
\allowdisplaybreaks
\subsection{The first-order condition of Eq.~(\ref{special_case})}
According to the first-order condition at $P^{*u}$, \eqref{welfare function under universal network allocation} yields:
\begin{align}
&\frac{\partial w_k(p_k^{*\!u})}{\partial p_k}=\!{c_o}\!\bigg(\!\frac{\partial \Hat{o}(P^{*\!u})}{\partial p_k} \frac{p_k^{*\!u}}{\sum_{t\in \mathcal{A}}p_t^{*\!u}} 
\!-\! \Hat{o}(P^{*\!u})\frac{\sum_{t\in \mathcal{A}}p_t^{*\!u}-p_k^{*\!u}}{(\sum_{t\in \mathcal{A}}p^{*\!u}_t)^2}\bigg)\nonumber\\
\!
&\!+\!{c_v}\!\bigg(\!\sum_{v\in \mathcal{N}} \!\!\!\frac{\partial |\Hat{v}_n(P^{*\!u})|}{\partial p_k}\!\frac{p_k^{*\!u}}{\sum_{t\in \mathcal{A}}\!p_t^{*\!u}} \!-\! \sum_{n\in \mathcal{N}}\!|\Hat{v}_n\!(\!P^{*\!u}\!)|\frac{\sum_{t\in \mathcal{A}}\!p_t^{*\!u}-p_k^{*\!u}}{(\sum_{t\in \mathcal{A}}\!p_k^{*\!u})^2}\!\bigg) \nonumber\\
&\!+\!{c_s}\!\bigg(\sum_{l\in \mathcal{L}} \!\!\frac{\partial |\Hat{s}_l^f(P^{*\!u})|}{\partial p_k}\frac{p_k^{*\!u}}{\sum_{t\in \mathcal{A}}p_t^{*\!u}} \!-\! \sum_{l\in \mathcal{L}}\!|\Hat{s}_l^f\!(\!P^{*\!u}\!)|\frac{\sum_{t\in \mathcal{A}}\!p_t^{*\!u}-p_k^{*\!u}}{(\sum_{t\in \mathcal{A}}\!p_t^{*\!u})^2}\!\bigg)
\label{first order condition universal allocation}
\end{align}
while it holds that:
\begin{align}
\begin{split}
\frac{\partial c_v}{\partial v(P)}\frac{\partial v(P)}{\partial p_k}=0, \quad
\frac{\partial c_s}{\partial v(P)}\frac{\partial v(P)}{\partial p_k}=0.
\end{split}
\end{align}
\subsection{Causal relationship factors of Eq.~(\ref{update_welfare_ftn})}
Based on the Proposition~\ref{prop:causality based cost allocation}, causal relationship factors $\mathrm{\Phi}_{n,k},\mathrm{X}_{l,k}$ and $\mathrm{\Psi}_{k}$ can be defined as follows:
\begin{subequations}
    \begin{align}
        & \mathrm{\Phi}_{n,k} = \frac{1}{|v_n(P^0)|}\!\operatorname{Re}\!\bigg(\!\overline{v}_n(P^0)\frac{\partial v_n(P^0)}{\partial p_k}\bigg), \\
        & \mathrm{X}_{l,k} = \frac{1}{|s_l^\mathrm{f}(P^0)|}\!\operatorname{Re}\!\bigg(\!\overline{s}_l^\mathrm{f}(P^0)\frac{\partial s_l^\mathrm{f}(P^0)}{\partial p_k}\bigg), \\  
        & \mathrm{\Psi}_{k} = \operatorname{Re}\!\Bigg(\!\sum_{n=0}^{N}\!\frac{\partial \overline{v}_n(P^0)}{\partial p_k}\!\bigg(\!\sum_{m=0}^{N}\!{G}_{nm}v_m(P^0)\!\bigg)\!\Bigg).
    \end{align}
\end{subequations}

\bibliographystyle{IEEEtran}
\bibliography{ref}

\begin{thebibliography}{10}
\providecommand{\url}[1]{#1}
\csname url@samestyle\endcsname
\providecommand{\newblock}{\relax}
\providecommand{\bibinfo}[2]{#2}
\providecommand{\BIBentrySTDinterwordspacing}{\spaceskip=0pt\relax}
\providecommand{\BIBentryALTinterwordstretchfactor}{4}
\providecommand{\BIBentryALTinterwordspacing}{\spaceskip=\fontdimen2\font plus
\BIBentryALTinterwordstretchfactor\fontdimen3\font minus
  \fontdimen4\font\relax}
\providecommand{\BIBforeignlanguage}[2]{{%
\expandafter\ifx\csname l@#1\endcsname\relax
\typeout{** WARNING: IEEEtran.bst: No hyphenation pattern has been}%
\typeout{** loaded for the language `#1'. Using the pattern for}%
\typeout{** the default language instead.}%
\else
\language=\csname l@#1\endcsname
\fi
#2}}
\providecommand{\BIBdecl}{\relax}
\BIBdecl

\bibitem{kim2023pricing}
H.~J. Kim \emph{et~al.}, ``Pricing mechanisms for peer-to-peer energy trading:
  Towards an integrated understanding of energy and network service pricing
  mechanisms,'' \emph{Renewable and Sustainable Energy Reviews}, vol. 183,
  2023.

\bibitem{tushar2021peer}
W.~Tushar, C.~Yuen, T.~K. Saha, T.~Morstyn, A.~C. Chapman, M.~J.~E. Alam,
  S.~Hanif, and H.~V. Poor, ``Peer-to-peer energy systems for connected
  communities: A review of recent advances and emerging challenges,''
  \emph{Applied Energy}, vol. 282, p. 116131, 2021.

\bibitem{tushar2019grid}
W.~Tushar, T.~K. Saha, C.~Yuen, T.~Morstyn, H.~V. Poor, R.~Bean \emph{et~al.},
  ``Grid influenced peer-to-peer energy trading,'' \emph{IEEE Transactions on
  Smart Grid}, vol.~11, no.~2, pp. 1407--1418, 2019.

\bibitem{kim2019direct}
H.~Kim, J.~Lee, S.~Bahrami, and V.~W. Wong, ``Direct energy trading of
  microgrids in distribution energy market,'' \emph{IEEE Transactions on Power
  Systems}, vol.~35, no.~1, pp. 639--651, 2019.

\bibitem{nguyen2021distributed}
D.~H. Nguyen and T.~Ishihara, ``Distributed peer-to-peer energy trading for
  residential fuel cell combined heat and power systems,'' \emph{International
  Journal of Electrical Power \& Energy Systems}, vol. 125, 2021.

\bibitem{anoh2019energy}
K.~Anoh, S.~Maharjan, A.~Ikpehai, Y.~Zhang, and B.~Adebisi, ``Energy
  peer-to-peer trading in virtual microgrids in smart grids: A game-theoretic
  approach,'' \emph{IEEE Transactions on Smart Grid}, vol.~11, no.~2, pp.
  1264--1275, 2019.

\bibitem{paudel2020peer}
A.~Paudel, L.~P. M.~I. Sampath, J.~Yang, and H.~B. Gooi, ``Peer-to-peer energy
  trading in smart grid considering power losses and network fees,'' \emph{IEEE
  Transactions on Smart Grid}, vol.~11, no.~6, pp. 4727--4737, 2020.

\bibitem{baroche2019exogenous}
T.~Baroche, P.~Pinson, R.~L.~G. Latimier, and H.~B. Ahmed, ``Exogenous cost
  allocation in peer-to-peer electricity markets,'' \emph{IEEE Transactions on
  Power Systems}, vol.~34, no.~4, pp. 2553--2564, 2019.

\bibitem{guerrero2018decentralized}
J.~Guerrero, A.~C. Chapman, and G.~Verbi{\v{c}}, ``Decentralized p2p energy
  trading under network constraints in a low-voltage network,'' \emph{IEEE
  Transactions on Smart Grid}, vol.~10, no.~5, pp. 5163--5173, 2018.

\bibitem{haggi2021multi}
H.~Haggi and W.~Sun, ``Multi-round double auction-enabled peer-to-peer energy
  exchange in active distribution networks,'' \emph{IEEE Trans. Smart Grid},
  vol.~12, 2021.

\bibitem{morstyn2019integrating}
T.~Morstyn, A.~Teytelboym, C.~Hepburn, and M.~D. McCulloch, ``Integrating p2p
  energy trading with probabilistic distribution locational marginal pricing,''
  \emph{IEEE Transactions on Smart Grid}, vol.~11, no.~4, pp. 3095--3106, 2019.

\bibitem{kim2019p2p}
J.~Kim and Y.~Dvorkin, ``A p2p-dominant distribution system architecture,''
  \emph{IEEE Trans. Power Syst.}, 2019.

\bibitem{maser2011s}
G.~Maser, ``It's electric, but ferc's cost-causation boogie-woogie fails to
  justify socialized costs for renewable transmission,'' \emph{Geo. LJ}, vol.
  100, p. 1829, 2011.

\bibitem{samadi2010optimal}
P.~Samadi \emph{et~al.}, ``Optimal real-time pricing algorithm based on utility
  maximization for smart grid,'' in \emph{2010 First IEEE international
  conference on smart grid communications}.\hskip 1em plus 0.5em minus
  0.4em\relax IEEE, 2010.

\bibitem{rosen1965existence}
J.~B. Rosen, ``Existence and uniqueness of equilibrium points for concave
  n-person games,'' \emph{Econometrica: Journal of the Econometric Society},
  pp. 520--534, 1965.

\bibitem{christakou2013efficient}
K.~Christakou, J.-Y. LeBoudec, M.~Paolone, and D.-C. Tomozei, ``Efficient
  computation of sensitivity coefficients of node voltages and line currents in
  unbalanced radial electrical distribution networks,'' \emph{IEEE Transactions
  on Smart Grid}, vol.~4, no.~2, pp. 741--750, 2013.

\bibitem{zhou2008simplified}
Q.~Zhou and J.~Bialek, ``Simplified calculation of voltage and loss sensitivity
  factors in distribution networks,'' in \emph{Proc. 16th Power Syst. Comput.
  Conf.(PSCC2008)}, 2008.

\bibitem{yan2020distribution}
M.~Yan, M.~Shahidehpour, A.~Paaso, L.~Zhang, A.~Alabdulwahab, and A.~Abusorrah,
  ``Distribution network-constrained optimization of peer-to-peer transactive
  energy trading among multi-microgrids,'' \emph{IEEE transactions on smart
  grid}, vol.~12, no.~2, pp. 1033--1047, 2020.

\bibitem{dolatabadi2020enhanced}
S.~H. Dolatabadi, M.~Ghorbanian, P.~Siano, and N.~D. Hatziargyriou, ``An
  enhanced {IEEE} 33 bus benchmark test system for distribution system
  studies,'' \emph{IEEE Transactions on Power Systems}, vol.~36, no.~3, pp.
  2565--2572, 2020.

\bibitem{codeanddata}
H.~J. Kim \emph{et~al.}, ``Code supplement for causality-based cost allocation
  for p2p energy trading in distribution system,'' 2023,
  \url{https://github.com/githjkim/causality_cost_allocation.git}.

\end{thebibliography}

\end{document}